\documentclass[11pt,draftcls,onecolumn]{IEEEtranwc}

\usepackage{color}
\usepackage{amsmath,amssymb,amscd}
\usepackage{subfigure, epsfig}
\newtheorem{lemma}{Lemma}
\newtheorem{theorem}{Theorem}

\newcommand{\mb}{\mathbf}
\newcommand{\ul}{\underline}
\newcommand{\ol}{\overline}
\newcommand{\bs}{\boldsymbol}
\newcommand{\mc}{\mathcal}

\newcommand{\ds}{\displaystyle}

\ifCLASSINFOpdf
\else

\begin{document}

\title{Power Allocation Games in Wireless Networks of Multi-antenna
Terminals}

\author{
Elena-Veronica~Belmega, Samson~Lasaulce, M\'erouane~Debbah,
Marc~Jungers, and Julien~Dumont

\thanks{E.~V. Belmega and S.~Lasaulce are with LSS (joint lab of CNRS, Sup\'{e}lec, Univ. Paris-Sud 11), Sup\'{e}lec,
Plateau du Moulon, 91192 Gif-sur-Yvette Cedex, France, \{belmega,lasaulce\}@lss.supelec.fr; M. Debbah is with the Alcatel-Lucent Chair on
Flexible Radio, Sup\'{e}lec, Plateau du Moulon, 91192 Gif-sur-Yvette Cedex, France, merouane.debbah@supelec.fr; M. Jungers is with CRAN,
Nancy-Universit\'e, CNRS, 2 avenue de la For\^et de Haye, 54516 Vandoeuvre-les-Nancy, France, marc.jungers@cran.uhp-nancy.fr; J. Dumont is with
Lyc\'ee Chaptal, 45 Boulevard des Batignolles, 75008 Paris, France, dumont@crans.org}}

\maketitle

\begin{abstract}
 We consider wireless networks that can be modeled by
multiple access channels in which all the terminals are equipped
with multiple antennas. The propagation model used to account for
the effects of transmit and receive antenna correlations is the
unitary-invariant-unitary model, which is one of the most general
models available in the literature. In this context, we introduce and
analyze two resource allocation games. In both games, the mobile
stations selfishly choose their power allocation policies in order
to maximize their individual uplink transmission rates; in
particular they can ignore some specified centralized policies. In
the first game considered, the base station implements successive
interference cancellation (SIC) and each mobile station chooses his
best space-time power allocation scheme; here, a coordination
mechanism is used to indicate to the users the order in which the
receiver applies SIC. In the second framework, the base
station is assumed to implement single-user decoding. For these two
games a thorough analysis of the Nash equilibrium is provided: the
existence and uniqueness issues are addressed; the corresponding
power allocation policies are determined by exploiting random matrix
theory; the sum-rate efficiency of the equilibrium is studied
analytically in the low and high signal-to-noise ratio regimes and by
simulations in more typical scenarios. Simulations show that, in
particular, the sum-rate efficiency is high for the type of systems
investigated and the performance loss due to the use of the proposed
suboptimum coordination mechanism is very small.

\end{abstract}


\section{Introduction} \label{sec:intro}

In this paper, we consider the uplink of a decentralized network of
several mobile stations (MS) and one base station (BS). This type of
network is commonly referred to as the decentralized multiple access
channel (MAC). The network is said to be decentralized in the sense
that each user can freely choose his power allocation (PA) policy in
order to selfishly maximize a certain individual performance
criterion, which is called utility or payoff. This means that, even
if the the BS broadcasts some specified policies, every user is free
to ignore the policy intended for him if the latter does not
maximize his performance criterion.

To the best of the authors' knowledge, the problem of decentralized
PA in wireless networks has been properly formalized for the first
time in \cite{grandhi-allerton-1992, grandhi-comm-1994}.
Interestingly, this problem can be formulated quite naturally as a
non-cooperative game with  different performance criteria
(utilities) such as the carrier-to-interference ratio
\cite{ji-wn-1998}, aggregate throughput \cite{oh-infocom-2000} or
energy efficiency \cite{goodman-pcomm-2000},
\cite{meshkati-jsac-2006}. In this paper, we assume that the users
want to maximize information-theoretic utilities and more precisely
their Shannon transmission rates. Indeed, the point of view adopted
here is close to the one proposed by the authors of
\cite{yu-jsac-2002} for DSL (digital subscriber lines) systems,
which are modeled as a parallel interference channel;
\cite{lai-it-2008} for the single input single output (SISO) and
single input multiple output (SIMO) fast fading MACs with global
CSIR and global CSIT (Channel State Information at the
Receiver/Transmitters); \cite{lasaulce-gamecomm-2007} for MIMO
(Multiple Input Multiple Output) MACs with global CSIR, channel
distribution information at the transmitters (global CDIT) and
single-user decoding (SUD) at the receivers; \cite{arslan-wc-2007,
scutari-jsac-2008} for Gaussian MIMO interference channels with
global CSIR and local CSIT and, by definition of the conventional
interference channel \cite{carleial-it-1978}, SUD at the receivers.
Note that reference \cite{palomar-it-2003} where the authors
considered Gaussian MIMO MACs with neither CSIT nor CDIT differs
from our approach and that of \cite{yu-jsac-2002, lai-it-2008,
lasaulce-gamecomm-2007, arslan-wc-2007, scutari-jsac-2008} because
in \cite{palomar-it-2003} the MIMO MAC is seen as a two-player
zero-sum game where the first player is the group of transmitters
and the second player is the set of MIMO sub-channels. The closest
works to the work presented here are \cite{lasaulce-gamecomm-2007}
and \cite{belmega-wnc3-2008}. Although this paper is in part based
on these works, it still provides significant contributions w.r.t.
to them, as explained below.

In \cite{lasaulce-gamecomm-2007}, the authors consider MIMO multiple access channels and assume SUD at the BS; the authors formulate the PA
problem into a team game in which each user chooses his PA to maximize the network sum-rate. In \cite{belmega-wnc3-2008}, the same type of
decentralized networks is considered but SIC is assumed at the BS. As each user needs to know his decoding rank in order to adapt his PA policy
to maximize his individual transmission rate, a coordination mechanism has to be introduced: the coordination signal precisely indicates to all
the users the decoding order used by the receiver. The present paper differs from these two contributions on at least four important technical
points: (i) when SUD is assumed, the PA game is not formulated as a team game but as a non-cooperative one; (ii) we exploit several proof
techniques that are different from \cite{lasaulce-gamecomm-2007}; (iii) while \cite{lasaulce-gamecomm-2007} and \cite{belmega-wnc3-2008} assume
a Kronecker propagation model with \emph{common} receive correlation we assume here a more general model, the unitary-invariant-unitary (UIU)
propagation model introduced by \cite{tulino-it-2005}, for which the users can have different receive antenna correlation profiles. This is
useful in practice since, for instance, it allows one to study propagation scenarios where some users can be in line of sight with the BS (the
receive antenna are strongly correlated) whereas other users can be surrounded by many obstacles, which can strongly decorrelate the receive
antennas for these users; (iv) while the authors of \cite{belmega-wnc3-2008} restricted their attention to either a purely spatial PA problem or
a purely temporal PA problem, we tackle here the general space-time PA problem.

In this context, our main objective is to study the equilibrium of
two power allocation games associated with the two types of decoding
schemes aforementioned (namely SIC and SUD). The motivation for this
is that the existence of an equilibrium allows network designers to
predict, with a certain degree of stability, the effective operating
state(s) of the network. Clearly, in our context, uniqueness is a
desirable feature of the equilibrium. As it will be seen, it is
possible to prove the existence in both games under investigation.
Uniqueness is proven in the case of SUD while it is conjectured for
the case of SIC. In order to establish the corresponding results,
the paper is structured as follows. After presenting the general
system model in Sec. \ref{sec:system-model-mimo}, we analyze in
detail the space-time PA game when SIC and a corresponding
coordination mechanism are assumed (Sec. \ref{sec:SIC}). For this
game, the existence and uniqueness of the NE are proven and the
equilibrium is determined by exploiting random matrix theory when
the numbers of antennas are sufficiently large. Its sum-rate
efficiency is also analyzed. In Sec. \ref{sec:SUD}, we analyze the
case of SUD since this decoding scheme, although suboptimal in terms
of performance (even in the case of a network with single-antenna
terminals), has some features that can be found desirable in some
contexts: the receiver complexity is low, there is no need for a
coordination signal, there is no propagation error since the data
flows are decoded in parallel and not successively and also it is
intrinsically fair. To analyze the case of the SUD-based PA game, we
will follow the same steps as in Sec. \ref{sec:SIC} and we will see
that, the equilibrium analysis can be deduced, to a large extent,
from the SIC case. Numerical results are provided in Sec.
\ref{sec:simulation-results} to illustrate our theoretical analysis
and to better assess the sum-rate efficiency of the considered
games. Sec. \ref{sec:conclusions} corresponds to the conclusion.


\section{System Model}
\label{sec:system-model-mimo}

We assume a MAC with arbitrary number of users, $K\geq 2$. Regarding the original definition of the MAC by \cite{wyner-it-1974} and
\cite{cover-book-1975}, the system under consideration has two common features: all transmitters send at once and at different rates over the
entire bandwidth, and the transmitters are using good codes in the sense of the Shannon rate. Our system differs from
\cite{wyner-it-1974}\cite{cover-book-1975} in the sense that multiple antennas are considered at the terminal nodes, channels vary over time and
the BS does not dictate the PA policies to the MSs. Also, we assume the existence of coordination signal which is perfectly known to all the
terminals. If the coordination signal is generated by the BS itself, this induces a certain cost in terms of downlink signaling but the
distribution of the coordination signal can then be optimized. On the other hand, if the coordination signal comes from an external source,
e.g., an FM transmitter, the MSs can acquire their coordination signal for free in terms of downlink signaling. However this generally involves
a certain sub-optimality in terms of uplink rate. In both cases, the coordination signal will be represented by a random variable denoted by $S
\in \mc{S}$. Since we study the $K-$user MAC, $\mc{S} = \{0,1,..., K!\}$ is a $K!+1$-element alphabet. When the realization is in $\{1,..., K!
\}$ , the BS applies SIC with a certain decoding order (game 1). When $S=0$ the BS always applies SUD (game 2), where all users are decoded
simultaneously (no interference cancellation). In a real wireless system the frequency at which the realizations would be drawn would be roughly
proportional to the reciprocal of the channel coherence time (i.e., $1/T_{\mathrm{coh}}$). Note that the proposed coordination mechanism is suboptimal because it does not depend on the realizations of the channel matrices. We will see that the corresponding
performance loss is in fact very small.

We will further consider that each mobile station is equipped with
$n_t$ antennas whereas the base station has $n_r$ antennas (thus we
assume the same number of transmitting antennas for all the users).
In our analysis, the flat fading channel matrices of the different
links vary from symbol vector (or space-time codeword) to symbol
vector. We assume that the receiver knows all the channel matrices
(CSIR) whereas each transmitter has only access to the statistics of
the different channels (CDIT). The equivalent baseband signal
received by the base station can be written as:
\begin{equation}\label{eq:system-model-mimo}
\ul{Y}^{(s)}(\tau)=\sum_{k=1}^K \bs{H}_k(\tau) \ul{X}_k^{(s)}(\tau)
+ \ul{Z}^{(s)}(\tau),
\end{equation}
where $\ul{X}_k^{(s)}(\tau)$ is the $n_t$-dimensional column vector of symbols transmitted by user $k$ at time $\tau$ for the realization $s \in
\mc{S}$ of the coordination signal, $\mb{H}_k(\tau) \in \mathbb{C}^{n_r \times n_t} $ is the channel matrix (stationary and ergodic process) of
user $k$ and $\ul{Z}^{(s)}(\tau)$ is a $n_r$-dimensional complex white Gaussian noise distributed as $\mathcal{N}(\ul{0}, \sigma^2
\mathbf{I}_{n_r})$. For the sake of clarity we will omit the time index $\tau$ from our notations.

In order to take into account the antenna correlation effects at the
transmitters and receiver, we will assume the different channel
matrices to be structured according to the
unitary-independent-unitary  model introduced in
\cite{tulino-it-2005}: 
\begin{equation}
\label{eq:uiu-model} \forall k \in \{1,...,K\}, \
\mb{H}_k=\mb{V}_k\tilde{\mb{H}}_k\mb{W}_k,
\end{equation}
 where $\mb{V}_k$
and $\mb{W}_k$ are deterministic unitary matrices that allow one to
take into consideration the correlation effects at the receiver and
transmitter. Also $\tilde{\mb{H}}_k$ is an $n_r \times n_t$ matrix
whose entries are zero-mean independent complex Gaussian random
variables with an arbitrary profile of variances, such that
$\mathbb{E}|\tilde{H}_k(i,j)|^2 = \frac{\sigma_k(i,j)}{n_t}$. The
Kronecker propagation model for which the channel transfer matrices
factorizes as
$\mb{H}_k=\mb{R}_k^{1/2}\tilde{\bs{\Theta}}_k\mb{T}_k^{1/2}$ is a
special case of the UIU model where the profile of variances is
separable i.e., $\mathbb{E}|\tilde{H}_k(i,j)|^2 =
\frac{d_k^{(\mathrm{R})}(i) d_k^{(\mathrm{T})}(j)}{n_t}$, with for
each $k$: $\bs{\Theta}_k$ is a random matrix with zero-mean i.i.d.
entries, $\mb{T}_k$ is the transmit antenna correlation matrix,
$\mb{R}_k$ is the receive antenna correlation matrix,
$\{d_k^{(\mathrm{T})}(j)\}_{j \in \{1,\hdots,n_t\}}$ and $\{d_k^{(\mathrm{R})}(i)\}_{i\in \{1,\hdots,n_r\} }$ are their
associated eigenvalues. In this paper we will consider that
$\mb{V}_k = \mb{V}$ for all users. The reason for assuming this will
be made clearer a little further. In spite of this simplification, we
will still be able to deal with some useful scenarios where the
users see different propagation conditions in terms of receive
antenna correlation.


\section{Successive Interference Cancellation}
\label{sec:SIC}


When SIC is assumed at the BS, the strategy of user $k \in \{1,2,...,K \}$, consists in choosing the best vector of precoding matrices $\mb{Q}_k
= \left(\mb{Q}_k^{(1)} , \mb{Q}_k^{(2)}, ..., \mb{Q}_k^{(K!)} \right)$ where $\mb{Q}_k^{(s)} = \mathbb{E} \left[\ul{X}_k^{(s)} \ul{X}_k^{(s),H}
\right]$, for $s \in \mc{S}$, in the sense of his utility function. For clarity sake, we will introduce another notation which will be used in the remaining of this section to replace the
realization $s$ of the coordination signal. We denote by $\mc{P}_K$ the set of all possible permutations of $K$
elements, such that $\pi \in \mc{P}_k$ denotes a certain decoding order for the $K$ users and $\pi(k)$ denotes the rank of user $k\in\mc{K}$ and
$\pi^{-1}\in \mc{P}_K$ denotes the inverse permutation (i.e. $\pi^{-1}(\pi(k))=k$) such that $\pi^{-1}(r)$ denotes the index of the user that is
decoded with rank $r\in \mc{K}$. We denote by $p_{\pi} \in [0,1]$ the probability that the receiver implements the decoding order $\pi \in
\mc{P}_K$, which means that $\ds{\sum_{\pi \in \mc{P}_K}} p_{\pi}=1$. At last note that there is a one-to-one mapping between the set of
realizations of the coordination signal $\mc{S}$ and the set of permutations $\mc{P}_K$, i.e. $\xi : \mc{S} \rightarrow \mc{P}_k$ such that
$\xi(\cdot)$ is a bijective function. This is the reason why the index $s$ can be replaced with the index $\pi$ without introducing any
ambiguity or loss of generality. The vector of precoding matrices can be denoted by $\mb{Q}=  \left(\mb{Q}_k^{(\pi)}\right)_{\pi \in
\mc{P}_K}$ and the utility function can be written as:

\begin{equation}
\label{eq:utility-mimo} u_k^{\mathrm{SIC}}(\mb{Q}_k, \mb{Q}_{-k}) =
\sum_{\pi \in \mc{P}_K}p_{\pi} R_k^{(\pi)}(\mb{Q}_k^{(\pi)},
\mb{Q}_{-k}^{(\pi)})
\end{equation}
where
\begin{equation}
\label{eq:mimo-st-rates} R_k^{(\pi)}(\mb{Q}_k^{(\pi)},
\mb{Q}_{-k}^{(\pi)}) = \mathbb{E} \log_2 \left|\mb{I}+\rho
\mb{H}_k\mb{Q}_k^{(\pi)}\mb{H}_k^H + \rho \ds{\sum_{\ell \in
\mc{K}_{k}^{(\pi)}}}
\mb{H}_{\ell}\mb{Q}_{\ell}^{(\pi)}\mb{H}_{\ell}^H \right| -
\mathbb{E} \log_2 \left|\mb{I}+ \rho \ds{\sum_{\ell \in
\mc{K}_{k}^{(\pi)}}}
\mb{H}_{\ell}\mb{Q}_{\ell}^{(\pi)}\mb{H}_{\ell}^H \right|
\end{equation}
with $\rho = \frac{1}{\sigma^2}$ and $\mc{K}_{k}^{(\pi)} =
\left\{\ell \in \mc{K}| \pi(\ell)\geq \pi(k) \right\}$ represents,
for a given decoding order $\pi$, the subset of users that will be
decoded after user $k$. Also, we use the standard notation $-k$ , which stands for the
other players than $k$. An important point to mention here is the
power constraint under which the utilities are maximized. Indeed for
user $k \in \{1,...,K\}$, the strategy set is defined as follows:
\begin{equation}
\begin{array}{lcl}  \mc{A}_k^{\mathrm{SIC}} = \left\{
\mb{Q}_k=\left(\mb{Q}_k^{(\pi)} \right)_{\pi \in \mc{P}_K} \right. &
| &
 \forall \pi \in \mc{P}_K, \mb{Q}_k^{(\pi)} \succeq
0,    \left. \ds{\sum_{\pi \in \mc{P}_K}} p_{\pi}
\mathrm{Tr}(\mb{Q}_k^{(\pi)})\leq n_t \ol{P}_k \right\}.
\end{array}
\end{equation}
In order to tackle the existence and uniqueness issues for Nash
equilibria in the general space-time PA game, we exploit and extend
the results from Rosen \cite{rosen-eco-1965}, which we will briefly
state here below in order to make this paper sufficiently
self-contained.

\begin{theorem}
\label{theorem-rosen-1} \cite{rosen-eco-1965} \emph{Let $\mc{G} =
(\mc{K}, \{\mc{A}_k\}_{k\in\mc{K}},\{u_k\}_{k\in\mc{K}}) $ be a game
where $\mc{K} = \{1,...,K\}$ is the set of players,
$\mc{A}_1,...,\mc{A}_K$ the corresponding sets of strategies and
$u_1,...,u_k$ the utilities of the different players. If the
following three conditions are satisfied: (i) each $u_k$ is
continuous in the all the strategies $\ul{a}_j \in \mc{A}_j, \forall
j \in \mc{K}$; (ii) each $u_k$ is concave in $\ul{a}_k \in
\mc{A}_k$; (iii) $\mc{A}_1, ..., \mc{A}_K$ are compact and convex
sets; then $\mc{G}$ has at least one NE.} \label{theo-1-rosen}
\end{theorem}

\begin{theorem}
\label{theorem-rosen-2} \cite{rosen-eco-1965} \emph{Consider the
$K$-player concave game of Theorem \ref{theo-1-rosen}. If the
following (diagonally strict concavity) condition is met: for all
$k\in\mc{K}$ and for all $(\ul{a}_k',\ul{a}_k'') \in \mc{A}_k^2$
such that there exists at least one index $j\in \mc{K}$ for which
$\ul{a}_j' \neq \ul{a}_j''$,
$\ds{\sum_{k=1}^K(\ul{a}_k''-\ul{a}_k')^{T}\left[\nabla_{\ul{a}_k}u_k(\ul{a}_k',\ul{a}_{-k}')
- \nabla_{\ul{a}_k}u_k(\ul{a}_k'',\ul{a}_{-k}'')\right]>0}$; then
the uniqueness of the NE is insured.} \label{theo-2-rosen}
\end{theorem}

In the space-time power allocation game under investigation, the obtained results are stated in the following theorem.

\begin{theorem}
\label{thm-st} \emph{[Existence of an NE] The joint space-time power allocation game described by: the set of players $k \in \mc{K}$; the sets
of actions $\mc{A}_k^{\mathrm{SIC}}$ and the utility functions $u_k^{\mathrm{SIC}}(\mb{Q}_k,\mb{Q}_{-k})$ given in (\ref{eq:utility-mimo}), has
a Nash equilibrium.}
\end{theorem}

\begin{proof}
It is quite easy to prove that the strategy sets
$\mc{A}_k^{\mathrm{SIC}}$ are convex and compact sets and that the
utility functions $u_k^{\mathrm{SIC}}(\mb{Q}_k,\mb{Q}_{-k})$ are
concave w.r.t. $\mb{Q}_k$ and continuous w.r.t. to
$(\mb{Q}_k,\mb{Q}_{-k})$ and by Theorem \ref{theorem-rosen-1} at
least one Nash equilibrium exists. For more details, the reader is
referred to Appendix \ref{appendix_1}.
\end{proof}

\begin{theorem} \emph{[Sufficient condition for uniqueness] If the
following condition is met}
\begin{equation}
\label{cond_ext_rosen_sic} \sum_{\pi \in \mc{P}_K} \sum_{k=1}^K
\mathrm{Tr} \left\{ (\mb{Q}_k^{(\pi)''}-\mb{Q}_k^{(\pi)'})
\left(\nabla_{Q_k^{(\pi)}}
u_k^{\mathrm{SIC}}(\mb{Q}_k',\mb{Q}_{-k}') -\nabla_{Q_k^{(\pi)}}
u_k^{\mathrm{SIC}} (\mb{Q}_k'',\mb{Q}_{-k}'') \right) \right\}
> 0
\end{equation}
\emph{for all $\mb{Q}_k' = \left(\mb{Q}_k^{(\pi)'} \right)_{\pi \in \mc{P}_K }, \mb{Q}_k'' = \left(\mb{Q}_k^{(\pi)''}\right)_{\pi \in \mc{P}_K }
\in \mc{A}_k^{\mathrm{SIC}}$ such that $(\mb{Q}_1',\hdots,\mb{Q}_K') \neq (\mb{Q}_1'',\hdots,\mb{Q}_K'')$, then the Nash
equilibrium in the power allocation game of Theorem \ref{thm-st} is unique.} \label{ext_rosen_sic}
\end{theorem}

This theorem corresponds to the matrix generalization of the diagonally strict concavity (DSC) condition of [17] and is proven in Appendix B. To
know whether this condition is verified or not in the MIMO MAC one needs to re-write it in a more exploitable manner. It can be checked that
$\mc{C}$ expresses as $\mc{C}= \ds{\sum_{\pi \in \mc{P}_K}} p_{\pi} \mc{T}_{\pi}$ where for each $\pi \in \mc{P}_K$, $\mc{T}_{\pi}$ is given by:
\begin{equation}
\begin{array}{lcl}
\mc{T}_{\pi} & = & \ds{\sum_{k=1}^K} \mathrm{Tr} \left\{ \left(\mb{Q}_k^{(\pi)''}-\mb{Q}_k^{(\pi)'} \right)
\left[\nabla_{Q_k^{(\pi)}}R_k^{(\pi)}(\mb{Q}_k^{(\pi)'},\mb{Q}_{-k}^{(\pi)'}) -\nabla_{Q_k^{(\pi)}}R_k^{(\pi)}
(\mb{Q}_k^{(\pi)''},\mb{Q}_{-k}^{(\pi)''}) \right] \right\} \\
& = & \ds{ \mathbb{E} \sum_{r=1}^K} \mathrm{Tr} \left\{ \rho \mb{H}_{\pi^{-1}(r)} (\mb{Q}_{\pi^{-1}(r)}^{(\pi)''}- \mb{Q}_{\pi^{-1}(r)}^{(\pi)'}
)
\mb{H}_{\pi^{-1}(r)}^H \right. \\
&& \left. \left[\left(\mb{I} + \rho \mb{H}_{\pi^{-1}(r)}\mb{Q}_{\pi^{-1}(r)}^{(\pi)'} \mb{H}_{\pi^{-1}(r)}^H + \rho \ds{\sum_{s=r+1}^K}
\mb{H}_{\pi^{-1}(s)}\mb{Q}_{\pi^{-1}(s)}^{(\pi)'} \mb{H}_{\pi^{-1}(s)} ^H \right)^{-1} - \right. \right. \\
& & \left. \left. \left(\mb{I} + \rho \mb{H}_{\pi^{-1}(r)}\mb{Q}_{\pi^{-1}(r)}^{(\pi)''} \mb{H}_{\pi^{-1}(r)} ^H + \rho \ds{\sum_{s=r+1}^K}
\mb{H}_{\pi^{-1}(s)}\mb{Q}_{\pi^{-1}(s)}^{(\pi)''} \mb{H}_{\pi^{-1}(s)} ^H\right)^{-1}\right] \right\} \\
& =& \ds{\mathbb{E} \sum_{r=1}^K} \mathrm{Tr} \left(\mb{A}_r^{(\pi)''}- \mb{A}_r^{(\pi)'}\right)\left[\left(\mb{I}+ \sum_{s=r}^K
\mb{A}_s^{(\pi)'}\right)^{-1}-\left(\mb{I}+ \sum_{s=r}^K
\mb{A}_s^{(\pi)''}\right)^{-1}\right] \\
& \triangleq & \mathbb{E}\left[F_{\pi}(\mb{H})\right]
\end{array}
\end{equation}
where $\mb{A}_r^{(\pi)'}=\rho \mb{H}_{\pi^{-1}(r)}\mb{Q}_{\pi^{-1}(r)}^{(\pi)'} \mb{H}_{\pi^{-1}(r)}^H$, $\mb{A}_r^{(\pi)''}=\rho
\mb{H}_{\pi^{-1}(r)}\mb{Q}_{\pi^{-1}(r)}^{(\pi)''} \mb{H}_{\pi^{-1}(r)}^H$ and the users have been ordered using their decoding rank rather than
their index. Notice that since the expectation operator is linear we can switch between the trace and expectation.


Let us denote by $\mb{H} = \left[\mb{H}_1,\hdots,\mb{H}_K\right]$, $\mb{Q}' = \left(\mb{Q}_k'\right)_{k \in \mc{K}}$, , $\mb{Q}'' =
\left(\mb{Q}_k''\right)_{k \in \mc{K}}$, $\mb{Q}^{(\pi)'} = \left(\mb{Q}_k^{(\pi)'}\right)_{k \in \mc{K}}$, , $\mb{Q}^{(\pi)''} =
\left(\mb{Q}_k^{(\pi)''}\right)_{k \in \mc{K}}$, $\mb{A}' = \left(\mb{A}_k^{(\pi)'}\right)_{\pi \in \mc{P}_K, k \in \mc{K}}$,  $\mb{A}''=
\left(\mb{A}_k^{(\pi)''}\right)_{\pi \in \mc{P}_K, k \in \mc{K}}$.

In order to prove that the DSC condition holds we have to prove that for all $\mb{Q}' \neq \mb{Q}''$ we have $\mc{C} > 0$.


Let us give a very useful result.

\begin{lemma}\emph{ For any positive definite matrices $\mb{A}_1$, $\mb{B}_1$, and any positive semi-definite matrices $\mb{A}_i$,
$\mb{B}_i$, $i \in \{2,\hdots,K\}$, we have that}
\begin{equation}
\label{eq:trace-ineq} \sum_{i=1}^{K} \mathrm{Tr} \left\{ \left(\mb{A}_i - \mb{B}_i \right) \left[ \left( \sum_{j=1}^{i} \mb{B}_j\right)^{-1}   -
\left( \sum_{j=1}^{i} \mb{A}_j\right)^{-1} \right] \right\} \geq 0
\end{equation}
where the equality holds if and only if $\mb{A}_j = \mb{B}_j$ for all $j \in \{1, \hdots, K\}$
\end{lemma}

The proof can be found in \cite{belmega-jipam-2009}, for $K=2$, and in \cite{belmega-ajmaa-2010} for arbitrary $K \geq 2$. Using this result, we
can prove that for any channel realization, any $\mb{Q}',\mb{Q}''$ and any $\pi \in \mc{P}_K$:

\begin{equation}
F_{\pi} (\mb{H}) = \mathrm{Tr} \left(\mb{A}_r^{(\pi)''}- \mb{A}_r^{(\pi)'}\right)\left[\left(\mb{I}+ \sum_{s=r}^K
\mb{A}_s^{(\pi)'}\right)^{-1}-\left(\mb{I}+ \sum_{s=r}^K \mb{A}_s^{(\pi)''}\right)^{-1}\right] \geq 0
\end{equation}
implying that $\mc{T}_{\pi} \geq 0$ and that $\mc{C} \geq 0$. Let us consider now two arbitrary covariance matrices such that $\mb{Q}' \neq
\mb{Q}''$. This means that there is at least one decoding order $\vartheta \in \mc{P}_K$ such that $\mb{Q}^{(\vartheta)'} \neq
\mb{Q}^{(\vartheta)''}$. We will prove that $\mc{T}_{\vartheta}> 0$ which will imply the desired result $\mc{C} > 0$.

\textbf{Remark:} Assuming that $\mathrm{rank} (\mb{H}_k^H\mb{H}_k) = n_t$, for all $k \in \mc{K}$, and $n_t \leq n_r + 1$, then $\mb{Q}' \neq
\mb{Q}''$ implies that $\mb{A}' \neq \mb{A}''$. This means that for any channel realization we have $F_{\vartheta} (\mb{H}) >0$ which implies
directly $\mc{T}_{\vartheta}> 0$ and $\mc{C}>0$.

For the general proof, let us define the following sets:
\begin{equation}
\begin{array}{lcl}
\mc{A}_H (\mb{Q}^{(\vartheta)'}, \mb{Q}^{(\vartheta)''}) & = & \left\{\mb{H} \in \mc{D}_H | \forall k \in \mc{K}: \mb{H}_k ({\mb{Q}_k}^{(\vartheta)'}-{\mb{Q}_k}^{(\vartheta)''}) \mb{H}_k^H = 0\right\} \\
\tilde{\mc{A}}_H (\mb{Q}^{(\vartheta)'}, \mb{Q}^{(\vartheta)''}) & = & \left\{\mb{H} \in \mc{D}_H | \exists k \in \mc{K}: \mb{H}_{k}
(\mb{Q}_k^{(\vartheta)'}-\mb{Q}_k^{(\vartheta)''}) \mb{H}_k^H \neq 0\right\}
\end{array}
\end{equation}
We know that:
\begin{equation}
\begin{array}{lcl}
\mc{T}_{\vartheta}& =& \ds{\mathbb{E} [F_{\vartheta} (\mb{H})]} \\
& = & \ds{\int_{\mc{D}_H} } F_{\vartheta} (\mb{H}) L(\mb{H}) \mathrm{d}\mb{H} \\
& = & \ds{\int_{\mc{A}_H (\mb{Q}^{(\vartheta)'}, \mb{Q}^{(\vartheta)''})}}  F_{\vartheta} (\mb{H}) L(\mb{H}) \mathrm{d}\mb{H}  + \ds{\int_{\tilde{\mc{A}}_H (\mb{Q}^{(\vartheta)'}, \mb{Q}^{(\vartheta)''})}  } F_{\vartheta} (\mb{H}) L(\mb{H}) \mathrm{d}\mb{H}  \\
& =  & \ds{\int_{\tilde{\mc{A}}_H (\mb{Q}^{(\vartheta)'}, \mb{Q}^{(\vartheta)''})} } F_{\vartheta} (\mb{H}) L(\mb{H}) \mathrm{d}\mb{H}
\end{array}
\end{equation}
where $L(\mb{H}) > 0$ stands for the p.d.f. of $\mb{H} \in \mc{D}_H \equiv \mathbb{C}^{n_r \times Kn_t}$. The second equality follows since
$\mc{D}_H =\mc{A}_H (\mb{Q}^{(\vartheta)'}, \mb{Q}^{(\vartheta)''}) \cup \tilde{\mc{A}}_H (\mb{Q}^{(\vartheta)'}, \mb{Q}^{(\vartheta)''})$. The
third equality follows since for all $\mb{H} \in \mc{A}_H (\mb{Q}^{(\vartheta)'}, \mb{Q}^{(\vartheta)''})$ we have that $F_{\vartheta} (\mb{H})
= 0$ from Lemma \ref{eq:trace-ineq}. We know that for all $\mb{H} \in \tilde{\mc{A}}_H (\mb{Q}^{(\vartheta)'}, \mb{Q}^{(\vartheta)''})$ we have
that $F_{\vartheta} (\mb{H}) > 0$. It suffices to prove that $\tilde{\mc{A}}_H (\mb{Q}^{(\vartheta)'}, \mb{Q}^{(\vartheta)''})$ is a subset of
non-zero Lebesgue measure to imply that $\mc{T}_{\vartheta}> 0$ and thus that $\mc{C}> 0$. It turns out that we can prove the existence of a
compact set $\mc{U}_H \subseteq \tilde{\mc{A}}_H(\mb{Q}^{(\vartheta)'}, \mb{Q}^{(\vartheta)''})$ for arbitrary $\mb{Q}^{(\vartheta)'} \neq
\mb{Q}^{(\vartheta)''}$.
 Thus, we have the desired result $\mc{C}>0$.

\emph{Determination of the Nash equilibrium.} In order to find the
optimal covariance matrices, we proceed in the same way as described
in \cite{lasaulce-gamecomm-2007}. First we will focus on the optimal
eigenvectors and then we will determine the optimal eigenvalues by
approximating the utility functions under the large system
assumption.
\begin{theorem}\emph{[Optimal eigenvectors] For all $k\in\mc{K}$, $\mb{Q}_k\in \mc{A}_k^{\mathrm{SIC}}$ there
is no loss of optimality by imposing the structure $\mb{Q}_k=( {\mb{Q}_k}^{(\pi)})_{\pi \in \mc{P}_K}$, $\mb{Q}_k^{(\pi)}=\mb{W}_k
{\mb{P}_k}^{(\pi)} {\mb{W}_k}^H$, in the sense that:
$$ \ds{
\max_{\mb{Q}_k\in\mc{A}_k^{\mathrm{SIC}}}
u_k^{\mathrm{SIC}}(\mb{Q}_k,\mb{Q}_{-k}) =
\max_{\mb{Q}_k\in\mc{S}_k^{\mathrm{SIC}}}
u_k^{\mathrm{SIC}}(\mb{Q}_k,\mb{Q}_{-k})},
$$
where $\mc{S}_k^{\mathrm{SIC}}= \left\{\mb{Q}_k= ({\mb{Q}_k}^{(\pi)})_{\pi \in \mc{P}_k} \in \mc{A}_k^{\mathrm{SIC}} | \mb{Q}_k^{(\pi)} =
\mb{W}_k \mb{P}_k^{(\pi)} \mb{W}_k^H  \right\}$, $s\in \mc{S}$, model from (\ref{eq:uiu-model}) and $\mb{P}_k^{(s)}=
\mathrm{Diag}(P_k^{(\pi)}(1),\hdots,P_k^{(\pi)}(n_t))$.}
\end{theorem}
The detailed proof of this result is given in Appendix \ref{appendix_5}. This result, although easy to obtain, it is instrumental in our context for
two reasons. First, the search of the optimum precoding matrices
boils down to the search of the eigenvalues of these matrices.
Second, as the optimum eigenvectors are known, available results in
random matrix theory can be exploited to find an accurate
approximation of these eigenvalues. Indeed, the eigenvalues are not
easy to find in the finite setting. They might be found using numerical techniques based on extensive search. Here, our approach consists in
approximating the utilities in order to obtain expressions which are not only
easier to interpret but also easier to be optimized w.r.t. the
eigenvalues of the precoding matrices. The key idea is to
approximate the different transmission rates by their large-system
equivalent in the regime of large number of antennas. The
corresponding approximates can be found to be accurate even for
relatively small number of antennas (see e.g.,
\cite{biglieri-issta-02}\cite{dumont-globecom-2006} for more
details).

Since we have assumed $\mb{V}_k=\mb{V}$, we can exploit the results
in \cite{tulino-book-04}\cite{tulino-it-2005} for single-user
MIMO channels, assuming the asymptotic regime in terms of the number
of antennas: $n_r \rightarrow \infty$, $n_t \rightarrow \infty$,
$\frac{n_r}{n_t} \rightarrow \beta$. The corresponding approximated
utility for user $k$ is:

\begin{equation}
\tilde{u}_k^{\mathrm{SIC}} (\{\mb{P}_k^{(\pi)}\}_{k \in \mc{K}, \pi
\in \mc{P}_K}) =\ds{\sum_{\pi \in \mc{P}_K}} p_{\pi}
\tilde{R}_k^{(\pi)}(\mb{P}_k^{(\pi)},\mb{P}_{-k}^{(\pi)})
\end{equation}
where
\begin{equation}
\label{eq:mimo-s-approx-sic}
\begin{array}{lcl}
\tilde{R}_k^{(\pi)}(\mb{P}_k^{(\pi)},\mb{P}_{-k}^{(\pi)})&=& \ds{
\frac{1}{n_r}\sum_{\ell \in \mc{K}_{k}^{(\pi)}\cup \{k\}}
\sum_{j=1}^{n_t} \log_2\left(1+(N_k^{(\pi)}+1)\rho
P_{\ell}^{(\pi)}(j)\gamma_{\ell}^{(\pi)}(j)\right) +}\\
& & \ds{ \frac{1}{n_r}\sum_{i=1}^{n_r} \log_2 \left(1+
\frac{1}{(N_k^{(\pi)}+1)n_t}\sum_{\ell \in \mc{K}_k^{(\pi)}\cup
\{k\}} \sum_{j=1}^{n_t}
\sigma_{\ell}(i,j)\delta_{\ell}^{(\pi)}(j) \right) - }\\
& & \ds{ \frac{1}{n_r} \sum_{\ell \in \mc{K}_k^{(\pi)}\cup \{k\} }
\sum_{j=1}^{n_t} \gamma_{\ell}^{(\pi)}(j)\delta_{\ell}^{(\pi)}(j) \log_2e
-}\\
& &\ds{ \frac{1}{n_r} \sum_{\ell \in \mc{K}_k^{(\pi)}}\sum_{j=1}^{n_t} \log_2\left(1+N_k^{(\pi)}\rho  P_{\ell}^{(\pi)}(j) \phi_{\ell}^{(\pi)}(j)\right) -} \\
& & \ds{\frac{1}{n_r} \sum_{i=1}^{n_r} \log_2 \left(1+\frac{1}{N_k^{(\pi)}n_t}\ds{\sum_{\ell\in \mc{K}_k^{(\pi)}}\sum_{j=1}^{n_t}\sigma_{\ell}(i,j)\psi_{\ell}^{(\pi)}(j)}\right)+ }\\
& &  \frac{1}{n_r} \ds{\sum_{\ell \in
\mc{K}_k^{(\pi)}}\sum_{j=1}^{n_r}
\phi_{\ell}^{(\pi)}(j)\psi_{\ell}^{(\pi)}(j)\log_2 e}
\end{array}
\end{equation}
where $N_k^{(\pi)} = |\mc{K}_k^{(\pi)}|$ and the parameters
$\gamma_k^{(\pi)}(j)$ and $\delta_k^{(\pi)}(j)$ $\forall j \in
\{1,\hdots, n_t \}$, $k \in \mc{K}$, $\pi \in \mc{P}_K$ are the
solutions of:
\begin{equation}
\label{alpha_sys_sic} \left\{
\begin{array}{lcl}
&&\forall j \in \{1,\hdots,n_t\}, \ell \in \mc{K}_k^{(\pi)}\cup\{k\}: \\
\gamma_{\ell}^{(\pi)}(j) & = & \ds{\frac{1}{(N_k^{(\pi)}+1)n_t}
\sum_{i=1}^{n_r}
\frac{\sigma_{\ell}(i,j)}{1+\frac{1}{(N_k^{(\pi)}+1)n_t}\ds{\sum_{r \in \mc{K}_k^{(\pi)}\cup \{k\}}\sum_{m=1}^{n_t}}\sigma_r(i,m)\delta_{r}^{(\pi)}(m)}}\\
\ds{\delta_{\ell}^{(\pi)}(j)} &=& \ds{\frac{(N_k^{(\pi)}+1)\rho
P_{\ell}^{(\pi)}(j) }{1+(N_k^{(\pi)}+1)\rho
P_{\ell}^{(\pi)}(j)\gamma_{\ell}^{(\pi)}(j)}},
\end{array}\right.
\end{equation}
and $\phi_{\ell}^{(\pi)}(j)$, $\psi_{\ell}^{(\pi)}(j)$, $\forall j
\in \{1,\hdots, n_t \}$  and $\pi \in \mc{P}_K$ are the unique
solutions of the following system:
\begin{equation}
\label{alpha_sys5} \left\{
\begin{array}{lcl}
&&\forall j \in \{1,\hdots,n_t\}, \ell \in \mc{K}_k^{(\pi)}: \\
\phi_{\ell}^{(\pi)}(j) & = & \ds{\frac{1}{N_k^{(\pi)}
n_t}\sum_{i=1}^{n_r}
\frac{\sigma_{\ell}(i,j)}{1+\frac{1}{N_k^{(\pi)}n_t}\ds{\sum_{r\in \mc{K}_k^{(\pi)}}\sum_{m=1}^{n_t}\sigma_{r}(i,m)\psi_{r}^{(\pi)}(m)}}}\\
\ds{\psi_{\ell}^{(\pi)}(j)} &=& \ds{\frac{N_k^{(\pi)}\rho
P_{\ell}^{\pi)}(j) }{1+N_k^{(\pi)}\rho
P_{\ell}^{(\pi)}(j)\phi_{\ell}^{(\pi)}(j)}}.
\end{array}\right.
\end{equation}

The corresponding water-filling solution is:

\begin{equation}
\label{waterfill} 
 P_k^{(\pi),\mathrm{NE}}(j) = \left[ \frac{1}{ \ln 2
n_r \lambda_k} - \frac{1}{N_k^{(\pi)} \rho \gamma_k^{(\pi)}(j)
}\right]^+,
\end{equation}
where $\lambda_k \geq 0$ is the Lagrangian multiplier tuned in order
to meet the power constraint: $$\ds{ \sum_{\pi \in
\mc{P}_K}\sum_{j=1}^{n_t} p_{\pi}\left[ \frac{1}{ \ln 2 n_r
\lambda_k} - \frac{1}{N_k^{(\pi)} \rho \gamma_k^{(\pi)}(j)
}\right]^+ }= n_t \ol{P}_k.$$ Note that to solve the system of
equations given above, we can use the same iterative power
allocation algorithm as the one described in
\cite{lasaulce-gamecomm-2007}.

At this point, an important point has to be mentioned. The existence
and uniqueness issues have be analyzed in the finite setting (exact
game) whereas the determination of the NE is performed in the
asymptotic regime (approximated game). It turns out that large
system approximates of ergodic transmission rates have the same
properties as their exact counterparts, as shown recently by
\cite{dumont-arxiv-2007}, which therefore ensures the existence and
uniqueness of the NE in the approximated game.

\emph{Nash Equilibrium efficiency.}  In order to measure the
efficiency of the decentralized network w.r.t. its centralized
counterpart we introduce the following quantity:
\begin{equation}
\mathrm{SRE} =
\frac{R_{\mathrm{sum}}^{\mathrm{NE}}}{C_{\mathrm{sum}}} \leq 1,
\end{equation}
where SRE stands for sum-rate efficiency; the quantity
$R_{\mathrm{sum}}^{\mathrm{NE}}$ represents the sum-rate of the
decentralized network at the Nash equilibrium, which is achieved for
certain choices of coding and decoding strategies; the quantity
$C_{\mathrm{sum}}$ corresponds to the sum-capacity of the
centralized network, which is reached only if the optimum coding and
decoding schemes are known. Note that this is the case for the MAC
but not for other channels like the interference channel. Obviously,
the efficiency measure we introduce here is strongly connected to
the price of anarchy \cite{roughgarden-jacm-2002} (POA). The difference
between $\mathrm{SRE}$ and $\mathrm{POA}$ is subtle. In our context,
information theory provides us with fundamental physical limits on
the social welfare (network sum-capacity) while in general no such
upper bound is available. In our case, the sum-capacity is given by:
\begin{equation}
C_{\mathrm{sum}} = \max_{(\bs{\Omega}_1, ..., \bs{\Omega}_K) \in
\mc{A}^{(C)}} \mathbb{E} \log \left| \mb{I} + \rho \sum_{k=1}^{K}
\mb{H}_k \bs{\Omega}_k \mb{H}_k^H \right|,
\end{equation} with
\begin{equation}
 \mc{A}^{(C)} = \left\{
(\bs{\Omega}_1, ..., \bs{\Omega}_K)  | \forall k \in \mc{K},
 \bs{\Omega}_k \succeq 0, \bs{\Omega}_k=\bs{\Omega}_k^H, \mathrm{Tr}(\bs{\Omega}_k) \leq n_t \ol{P}_k
\right\}.
\end{equation}
In general, it is not easy to find a closed-form expression of the
SRE. This is why we will respectively analyze the SRE in the regimes
of high and low signal-to-noise ratio (SNR), and for intermediate
regimes simulations will complete our analysis. It turns out that
the SRE tends to 1 in the two mentioned extreme regimes, which is
the purpose of what follows.

In the \emph{high SNR regime}, where $\rho \rightarrow \infty$, we
observe from (\ref{alpha_sys_sic}) that $\delta_{\ell}^{(\pi)}(j) \rightarrow \frac{1}{\gamma_{\ell}^{(\pi)}(j)}$. Under this condition, it is easy to check that
by setting the derivatives of $\mc{L}_k$ w.r.t. $P_k^{(s)}(j)$ to zero,
we obtain that the power allocation policy at the NE is the uniform
power allocation $\mb{P}_k^{(\pi),\mathrm{NE}} = \ol{P}_k \mb{I}$,
regardless the realization of the coordination signal $S$.
Furthermore, in the high SNR regime, the sum-capacity is achieved by
the uniform power allocation. Thus, we obtain that the gap between the NE
achievable sum-rate and the sum-capacity is optimal,
$\mathrm{SRE}=1$ for any distribution of $S$.

In the \emph{low SNR regime}, where $\rho \rightarrow 0$, from (\ref{alpha_sys_sic}) we obtain
that $\delta_{\ell}^{(\pi)}(j) \rightarrow 0$ and that $\gamma_{\ell}^{(\pi)}(j) = \frac{1}{(N_k^{(\pi)}+1)n_t} \ds{\sum_{i=1}^{n_r}\sigma_{\ell}(i,j)}$. By approximating $\ln(1+x)\approx x$ when $x<<1$, the power allocations policies at the NE are the solutions of the following linear
programs:
\begin{equation}
\begin{array}{ccl}
 \ds{\max_{\{P_k^{(\pi)}(j)\}_{1\leq j\leq n_t}}  \sum_{j=1}^{n_t}
 \left\{\sum_{\pi \in \mc{P}_K} p_{\pi}  P_k^{(\pi)}(j) \sum_{i=1}^{n_r}\sigma_k(i,j)\right\}}  \\
\text{\hspace{25pt} s.t.\hspace{5pt}} \ds{\sum_{j=1}^{n_t} \sum_{\pi
\in \mc{P}_K} P_k^{(\pi)}(j)\leq \ol{P}_k n_t}
\end{array},
\end{equation}
given by:
\begin{equation}
\sum_{\pi \in \mc{P}_K} p_{\pi} P_k^{(\pi), \mathrm{NE}}(j)= \left|
\begin{array}{cllc}
n_t\ol{P}_k & \ \ \text{if } j = \ds{\arg \max_{1\leq m \leq n_t} \sum_{i=1}^{n_r}\sigma_k(i,m)}\\
0 & \ \ \text{otherwise }
\end{array}
\right. .
\end{equation}
The optimal power allocation that achieves the sum-capacity is equal
to the equilibrium power allocation, $\mb{P}_k^{*}= \sum_{\pi\in
\mc{P}_K} p_{\pi} \mb{P}_k^{(\pi), \mathrm{NE}}(j)$ Thus, the achievable
sum-rate at the NE is equal to the centralized upper bound and thus
$\mathrm{SRE}=1$ for any distribution of $S$. In conclusion, when either the low or high SNR regime is assumed, the
sum-capacity of the fast fading MAC is achieved at the NE although a
sub-optimum coordination mechanism is assumed and also regardless of
the distribution of the coordination channel.


\section{Single User Decoding}
\label{sec:SUD}

In this section the coordination signal is deterministic (namely
$\mathrm{Pr}[S=s] = \delta(s)$, $\delta$ being the Kronecker symbol)
and therefore the amount of downlink signalling the BS needs in
order to indicate to the MSs that it is using SUD can be made
arbitrary small (by letting the frequency at which the realizations
of the coordination signal are drawn tend to zero). In this
framework, each user has to optimize only one precoding matrix.
Indeed, the strategy of user $k \in \mc{K}$, consists in choosing
the best precoding matrix $\mb{Q}_k^{(0)} = \mathbb{E}
\left[\ul{X}_k^{(0)} \ul{X}_k^{(0)H} \right]$, in the sense of his utility
function obtained with SUD:
\begin{equation}
\label{eq:utility-sud} u_k^{\mathrm{SUD}}(\mb{Q}_k^{(0)},
\mb{Q}_{-k}^{(0)}) = \mathbb{E} \log \left| \mb{I} + \rho \mb{H}_k
\mb{Q}_k^{(0)} \mb{H}_k^H + \rho\ds{\sum_{\ell \neq k}} \mb{H}
_{\ell} \mb{Q}_{\ell}^{(0)} \mb{H}_{\ell}^H \right|- \mathbb{E} \log
\left|\mb{I} + \rho \ds{\sum_{\ell \neq
k}}\mb{H}_{\ell}\mb{Q}_{\ell}^{(0)}\mb{H}_{\ell}^H \right|
\end{equation}. The strategy set of user $k$ becomes
\begin{equation}
\begin{array}{lcl}  \mc{A}_k^{\mathrm{SUD}} = \left\{
\mb{Q}_k^{(0)} \succeq 0, \mb{Q}_k^{(0)}=\mb{Q}_k^{(0),H}, \mathrm{Tr}(\mb{Q}_k^{(0)}) \leq n_t \ol{P}_k \right\}.
\end{array}
\end{equation}
It turns out that the equilibrium analysis in the game with SUD can
be, to a large extent, deduced from the game with SIC. For this
reason, we will not detail the corresponding proofs. The existence
and uniqueness issues are given in the following theorem.

\begin{theorem}
\label{thm-sud} \emph{[Existence and uniqueness of an NE] The space
power allocation game described by: the set of players $k \in
\mc{K}$; the sets of actions $\mc{A}_k^{\mathrm{SUD}}$ and the
payoff functions
$u_k^{\mathrm{SUD}}(\mb{Q}_k^{(0)},\mb{Q}_{-k}^{(0)})$ given in
(\ref{eq:utility-sud}), has a unique Nash equilibrium.}
\end{theorem}

To prove the \emph{existence} of a Nash equilibrium we also exploit
Theorem \ref{theorem-rosen-1} and the four necessary conditions on
the utility functions and strategy sets can be verified using the
same tools as described in Appendix \ref{appendix_1}.

\emph{Uniqueness of the Nash equilibrium.} Here we can specialize Theorem 4, which is the matrix extension of Theorem 2. When the strategies
sets are not sets of pairs of matrices but only sets of matrices, the diagonally strict concavity condition in (6) can be written as follows.
For all  $\mb{Q}_k^{(0)'}, \mb{Q}_k^{(0)''} \in \mc{A}_k^{\mathrm{SUD}}$ such that $(\mb{Q}_1^{(0)'},\hdots,\mb{Q}_K^{(0)'}) \neq
(\mb{Q}_1^{(0)''},\hdots,\mb{Q}_K^{(0)''})$:
\begin{equation}
\label{cond_ext_rosen_SUD}
\begin{array}{lcl}
\mc{C} & = & \ds{\sum_{k=1}^K} \mathrm{Tr} \left\{ (\mb{Q}_k^{(0)''}-\mb{Q}_k^{(0)'})
\left[\nabla_{Q_k^{(0)}}u_1(\mb{Q}_k^{(0)'},\mb{Q}_k^{(0)'}) -\nabla_{Q_k^{(0)}}u_1(\mb{Q}_k^{(0)''},\mb{Q}_k^{(0)''})\right] \right\}.
\end{array}
\end{equation}
Now we can evaluate $\mc{C}$ and obtain that:
\begin{equation}
\begin{array}{lcl}
\mc{C} &= &\mathbb{E} \ds{\sum_{k=1}^{K}} \mathrm{Tr}\left\{\left[\rho\mb{H}_k(\mb{Q}_k^{(0)'}-\mb{Q}_k^{(0)''}
)\mb{H}_k^H\right]\left[\left(\mb{I}+\rho \ds{\sum_{\ell =1 }^K}
\mb{H} _{\ell} \mb{Q}_{\ell}^{(0)''} \mb{H}_{\ell}^H \right)^{-1} \right. \right. \\
& & \left. \left. - \left(\mb{I}+\rho \ds{\sum_{\ell =1}^K} \mb{H}
_{\ell} \mb{Q}_{\ell}^{(0)'} \mb{H}_{\ell}^H \right)^{-1}\right]\right\} \\
& = & \mathbb{E} \mathrm{Tr}
\{(\mb{B}^{'}-\mb{B}^{''})[(\mb{B}^{''})^{-1}-(\mb{B}^{'})^{-1}]\}, \\
& = & \mathbb{E} [F_0(\mb{H})]
\end{array}
\end{equation}
which is positive for any $\mb{B}^{'} = \mb{I}+\ds{\sum_{\ell =1}^K} \mb{H}_{\ell}\mb{Q}^{(0)'} \mb{H}_{\ell}^H$, $\mb{B}^{''} =
\mb{I}+\ds{\sum_{\ell =1}^K }\mb{H}_{\ell}\mb{Q}^{(0)''} \mb{H}_{\ell}^H$ from (\ref{eq:trace-ineq}) for $K=2$. We need to prove that for any
$(\mb{Q}_1^{(0)'},\hdots,\mb{Q}_K^{(0)'}) \neq (\mb{Q}_1^{(0)''},\hdots,\mb{Q}_K^{(0)''})$ we have $\mc{C}> 0$.

\textbf{Remark:} Assuming that $\mathrm{rank} (\mb{H}^H\mb{H}) = Kn_t$ and $Kn_t \leq n_r + K$, then $\mb{Q}' \neq \mb{Q}''$ implies that
$\mb{B}' \neq \mb{B}''$. This means that for any channel realization we have $F_{0} (\mb{H}) >0$ which implies directly that $\mc{C}>0$.

For the general proof, we define the following sets:
\begin{equation}
\begin{array}{lcl}
\mc{B}_H (\mb{Q}^{(0)'}, \mb{Q}^{(0)''}) & = & \left\{\mb{H} \in \mc{D}_H | \sum_{k=1}^K \mb{H}_k ({\mb{Q}_k}^{(0)'}-{\mb{Q}_k}^{(0)''}) \mb{H}_k^H = 0\right\} \\
\tilde{\mc{B}}_H (\mb{Q}^{(0)'}, \mb{Q}^{(0)''}) & = & \left\{\mb{H} \in \mc{D}_H | \sum_{k=1}^K \mb{H}_k ({\mb{Q}_k}^{(0)'}-{\mb{Q}_k}^{(0)''})
\mb{H}_k^H \neq 0\right\}
\end{array}
\end{equation}

We know that:

\begin{equation}
\begin{array}{lcl}
\mc{C}& =& \ds{\mathbb{E} [F_{0} (\mb{H})]} \\
& = &\ds{ \int_{\mc{D}_H}}  F_0(\mb{H}) L(\mb{H}) \mathrm{d}\mb{H} \\
& = & \ds{\int_{\mc{B}_H (\mb{Q}^{(0)'}, \mb{Q}^{(0)''})}}  F_{0} (\mb{H}) L(\mb{H}) \mathrm{d}\mb{H}  + \ds{\int_{\tilde{\mc{B}}_H (\mb{Q}^{(0)'}, \mb{Q}^{(0)''})} } F_{0} (\mb{H}) L(\mb{H}) \mathrm{d}\mb{H}  \\
& =  & \ds{\int_{\tilde{\mc{B}}_H (\mb{Q}^{(0)'}, \mb{Q}^{(0)''})} } F_{0} (\mb{H}) L(\mb{H}) \mathrm{d}\mb{H}
\end{array}
\end{equation}

The second equality follows since $\mc{D}_H =\mc{B}_H (\mb{Q}^{(0)'}, \mb{Q}^{(0)''}) \cup \tilde{\mc{B}}_H (\mb{Q}^{(0)'}, \mb{Q}^{(0)''})$.
The third equality follows because $F_{0} (\mb{H}) = 0$ for all $\mb{H} \in \mc{B}_H (\mb{Q}^{(0)'}, \mb{Q}^{(0)''})$ from Lemma 1. We also know
that $F_{0} (\mb{H}) > 0$ for all $\mb{H} \in \tilde{\mc{B}}_H (\mb{Q}^{(0)'}, \mb{Q}^{(0)''})$. It suffices to prove that $\tilde{\mc{B}}_H
(\mb{Q}^{(0)'}, \mb{Q}^{(0)''})$ is a subset of non-zero Lebesgue measure to imply that  $\mc{C}> 0$. Here as well, the existence of the compact
set can be proved (similarly to the proof for the SIC decoding technique).

\emph{Determination of the Nash equilibrium.} As for the optimal
eigenvectors of the covariance matrices, we follow the same lines as
in Appendix \ref{appendix_5}. In this case also there is no loss of
optimality by choosing the covariance matrices
$\mb{Q}_k^{(0)}=\mb{W}_k\mb{P}_k^{(0)} \mb{W}_k^H$, where $\mb{W}_k$
is the same unitary matrix as in (\ref{eq:uiu-model}) and $\mb{P}_k$
is the diagonal matrix containing the eigenvalues of
$\mb{Q}_k^{(0)}$.

Here also we further exploit the asymptotic results for the MIMO
channel given in \cite{tulino-book-04} \cite{tulino-it-2005}. The
approximated utility for user $k$ is:
 \begin{equation}
\label{eq:mimo-s-approx-sud}
\begin{array}{lcl}
\tilde{u}_k^{\mathrm{SUD}}(\mb{P}_k^{(0)},\mb{P}_{-k}^{(0)})&=& \ds{
\frac{1}{n_r}\sum_{k=1}^K \sum_{j=1}^{n_t} \log_2(1+K\rho
P_k^{(0)}(j)\gamma_k(j)) +}\\
& & \ds{ \frac{1}{n_r}\sum_{i=1}^{n_r} \log_2 \left(1+
\frac{1}{Kn_t}\sum_{k=1}^K \sum_{j=1}^{n_t}
\sigma_k(i,j)\delta_k(j) \right) - }\\
& & \ds{ \frac{1}{n_r} \sum_{k=1}^K \sum_{j=1}^{n_t}
\gamma_k(j)\delta_k(j) \log_2e -}
\\
& &\ds{ \frac{1}{n_r} \sum_{\ell \neq k}\sum_{j=1}^{n_t} \log_2(1+(K-1)\rho  P_{\ell}^{(0)}(j) \phi_{\ell}(j)) -} \\
& & \ds{\frac{1}{n_r} \sum_{i=1}^{n_r} \log_2 \left(1+\frac{1}{(K-1)n_t}\ds{\sum_{\ell\neq k}\sum_{j=1}^{n_t}\sigma_{\ell}(i,j)\psi_{\ell}(j)}\right)+ }\\
& &  \frac{1}{n_r} \ds{\sum_{\ell \neq k}\sum_{j=1}^{n_r}
\phi_{\ell}(j)\psi_{\ell}(j)\log_2 e}
\end{array}
\end{equation}

where the parameters $\gamma_k(j)$ and $\delta_k(j)$ $\forall j \in
\{1,\hdots, n_t \}$, $k \in \{1,2\}$ are solution of:

\begin{equation}
\label{alpha_sys5} \left\{
\begin{array}{lcl}
&&\forall j \in \{1,\hdots,n_t\}, k \in \mc{K}: \\
\gamma_k(j) & = & \ds{\frac{1}{Kn_t}\sum_{i=1}^{n_r}
\frac{\sigma_k(i,j)}{1+\frac{1}{Kn_t}\ds{\sum_{\ell=1}^K\sum_{m=1}^{n_t}\sigma_{\ell}(i,m)\delta_{\ell}(m)}}}\\
\ds{\delta_k(j)} &=& \ds{\frac{K\rho P_k^{(0)}(j) }{1+K\rho
P_k^{(0)}(j)\gamma_k(j)}}.
\end{array}\right.
\end{equation}

 and $\phi_{\ell}(j)$, $\psi_{\ell}(j)$, $\forall j \in \{1,\hdots, n_t \}$ are the
unique solutions of the following system:

\begin{equation}
\label{alpha_sys5} \left\{
\begin{array}{lcl}
&&\forall j \in \{1,\hdots,n_t\}, \ell \in \mc{K}\setminus\{k\}: \\
\phi_{\ell}(j) & = & \ds{\frac{1}{(K-1)n_t}\sum_{i=1}^{n_r}
\frac{\sigma_{\ell}(i,j)}{1+\frac{1}{(K-1)n_t}\ds{\sum_{r\neq k}\sum_{m=1}^{n_t}\sigma_{r}(i,m)\psi_{r}(m)}}}\\
\ds{\psi_{\ell}(j)} &=& \ds{\frac{(K-1)\rho P_{\ell}^{(0)}(j)
}{1+(K-1)\rho P_{\ell}^{(0)}(j)\phi_{\ell}(j)}}.
\end{array}\right.
\end{equation}


The corresponding water-filling solution is:

\begin{equation}
\label{waterfill} 
 P_k^{(0),\mathrm{NE}}(j) = \left[ \frac{1}{ \ln 2
n_r \lambda_k} - \frac{1}{K \rho \gamma_k(j) }\right]^+,
\end{equation}
where $\lambda_k \geq 0$ is the Lagrangian multiplier tuned in order
to meet the power constraint: $\ds{ \sum_{j=1}^{n_t}\left[ \frac{1}{
\ln 2 n_r \lambda_k} - \frac{1}{K \rho \gamma_k(j) }\right]^+ }= n_t
\ol{P}_k$.

In what the efficiency of
the NE point is concerned, we already know that the SUD decoding technique is
sub-optimal in the centralized case (SUD does allow the network to
operate at an arbitrary point of the centralized MAC capacity
region) and it is impossible to reach the sum-capacity
$C_{\mathrm{sum}}$ even if the high and low SNR regime are assumed.

\section{SIMULATION RESULTS}
\label{sec:simulation-results} 
 In what follows, we assume the regime of large
numbers of antennas. From \cite{lasaulce-gamecomm-2007},
\cite{tulino-book-04},
 \cite{tulino-it-2005}, we know that the approximates of the ergodic achievable rates in the asymptotic regime are accurate even for relatively small
  number of antennas. For the channel matrices, we assume the Kronecker model
  $\mb{H}_k=\mb{R}_k^{1/2}\mb{\Theta}_k\mb{T}_k^{1/2}$ mentioned in Sec.
  \ref{sec:system-model-mimo},
where the receive and transmit correlation matrices $\mb{R}_k$,
$\mb{T}_k$ follow an exponential profile characterized by the
correlation coefficients (see e.g.,
\cite{chiani-it-2003,skupch-pwc-2005}) $r=[r_1, r_2]$ and $t=[t_1,
t_2]$ such that $\mb{R}_k(i,j)=r_k^{|i-j|}$,
$\mb{T}_k(i,j)=t_k^{|i-j|}$. By assuming that the receive antenna is
a uniform linear array (ULA) and knowing that, when the dimensions
of Toeplitz matrices increase they can be approximated by circular
matrices we obtain that all the receive correlation matrices
$\mb{R}_k$ can be diagonalized in the same vector basis (i.e., the
Fourier basis). Thus the considered model is included in the UIU
model that we studied where $\mb{V}_k=\mb{V}$.

\emph{Fair SIC decoding versus SUD decoding.} First we compare the
results of the general space-time PA game considered in Sec.
\ref{sec:SIC}, where SIC decoding is used at the receiver, and the
game described in Sec. \ref{sec:SUD}, where SUD decoding is used.
Fig. \ref{fig1} depicts the achievable sum-rate at the equilibrium
as a function of the transmit power $P_1=P_2=P$, for the scenario
$n_r=n_t=10$, $r=[0.5, 0.2]$, $t=[0.5, 0.2]$, $\rho=3 dB$. In order
to have a fair comparison we assume that $p=\frac{1}{2}$ (on average
each user is decoded second half of the time when SIC is assumed).
We observe that, even in this scenario, which was thought to be a
bad one in terms of sub-optimality, the sum-rate obtained with the
first game is very close to the sum-capacity upper bound. Also, the
sum-rate reached when the BS uses SUD is clearly much lower than the
sum-rate obtained by using SIC.

\emph{SIC decoding, comparison between the joint space-time PA and
the special cases of spatial PA and temporal PA.} Now we want to
compare the results of the general space-time PA with the two
particular cases that were studied in \cite{belmega-wnc3-2008}: the
spatial PA, where the users are forced to allocate their power
uniformly over time (regardless of their decoding rank) but are free
to allocate their power over the transmit antennas; the temporal PA,
where the users are forced to allocate their power uniformly over
their antennas but they can adjust their power as a function of the
decoding rank at the receiver. Fig. \ref{fig2} represents the
sum-rate efficiency as a function of the coordination signal
distribution parameter $p \in [0,1]$ when $n_r=n_t=10$, $r=[0.3,
0]$, $t=[0.5, 0.2]$, $\rho=4 dB$, $P_1=5$, $P_2=50$. We observe that
the three types of power allocation policies perform very close to
the upper bound. What is most interesting is the fact that the
performance of the network at the equilibrium is better by using a
purely spatial PA instead of the most general space-time PA. This
has been confirmed by many other simulations and illustrates a Braess paradox: although the sets of strategies for the
space-time case include those of the purely spatial case, the
performance obtained at the NE are not better in the space-time
case.

\emph{SIC decoding, spatial PA, achievable rate region.} In Fig.
\ref{fig3}, we observe that the rate region achieved at the NE of the
space PA as a function of the distribution of the coordination
signal $p$ for the scenario $n_r=n_t=10$, $r=[0.4, 0.2]$, $t=[0.6,
0.3]$, $\rho=3 dB$, $P_1=5$, $P_2=50$. It is quite remarkable that
in large MIMO MACs, the capacity region comprises a full cooperation
segment just like the SISO MACs. The coordination signal precisely
allows one to move along the corresponding line. This shows the
relevance of large systems in decentralized networks since they
allow to determine the capacity region of certain systems whereas it
is unknown in the finite setting. Furthermore, they induce an
averaging effect, which makes the users' behavior predictable.



\section{CONCLUSIONS}
\label{sec:conclusions}

Interestingly, the existence and uniqueness of the Nash equilibrium can be proven in multiple access channels with multi-antenna terminals for a
general propagation channel model (namely the unitary-invariant-unitary model) and the most general case of space-time power allocation schemes.
In particular, the uniqueness proof requires a matrix generalization of the second theorem of Rosen \cite{rosen-eco-1965} and proving a trace
inequality \cite{belmega-ajmaa-2010}. For all the types of power allocation policies (purely temporal PA, purely spatial PA, space-time PA), the
sum-rate efficiency of the decentralized network is close to one when SIC is assumed and the network is coordinated by the proposed suboptimum
coordination mechanism. Quite surprisingly, the space-time power allocation performs a little worse than its purely spatial counterpart, which
puts in evidence a Braess paradox in the types of wireless networks under consideration. One of the interesting extensions of this work would be
to analyze the impact of a non-perfect SIC on the PA problem. Indeed, the effect of propagation errors could then be assessed (which does not
exist with SUD).

\appendices

\section{}
\label{appendix_1}

\subsection{Concavity of the utility functions $u_k^{\mathrm{SIC}}$}

Let us focus on user $k \in \mc{K}$. We want to prove that $u_k^{\mathrm{SIC}}(\mb{Q}_k,\mb{Q}_{-k}) $ is concave w.r.t. $\mb{Q}_k \in
\mc{A}_1^{\mathrm{SIC}}$. We observe that the term $R_k^{(\pi)}(\mb{Q}_k^{(\pi)},\mb{Q}_{-k}^{(\pi)})$ in (\ref{eq:utility-mimo}) depends only
on $\mb{Q}_k^{(\pi)}$ and $\mb{Q}_{-k}^{(\pi)}$ and not on the covariance matrices $\mb{Q}_k^{(\tau)}$, $\mb{Q}_{-k}^{(\tau)}$ for any other
possible decoding rule $\tau \in \mc{P}_K \setminus \{\pi\}$. Thus, in order to prove that $u_k^{\mathrm{SIC}}(\mb{Q}_k,\mb{Q}_k)$ is strictly concave w.r.t.
to $\mb{Q}_k =(\mb{Q}_k^{(\pi)})_{\pi \in \mc{P}_K}$, it suffices to prove that $R_k^{(\pi)}(\mb{Q}_k^{(\pi)}, \mb{Q}_{-k}^{(\pi)})$ is
concave w.r.t. $\mb{Q}_k^{(\pi)}$ for all $\pi \in \mc{P}_K$.

To this end, we study the concavity of the function $f(\lambda)= R_k^{(\pi)}(\lambda \mb{Q}_k^{(\pi)'}+(1-\lambda) \mb{Q}_k^{(\pi)''}) $ over
the interval $[0,1]$ for any pair of matrices $(\mb{Q}_k^{(\pi)'},\mb{Q}_k^{(\pi)''})$. The second derivative of $f$ is equal to:
\[
\begin{array}{lcl}
\frac{\partial^2 f}{\partial \lambda^2}(\lambda)&  =  & - \mathbb{E} \mathrm{Tr} \left[ \rho^2 \mb{H}_k^H \left( \mb{I} + \rho \mb{H}_k
\mb{Q}_k^{(\pi)''} \mb{H}_k^H + \rho \lambda \mb{H}_k \Delta \mb{Q}_k^{(\pi)} \mb{H}_k^H + \rho \ds{\sum_{\ell \in \mc{K}_{k}^{(\pi)}}}
\mb{H}_{\ell}\mb{Q}_{\ell}^{(\pi)}\mb{H}_{\ell}^H
\right)^{-1} \mb{H}_k \Delta \mb{Q}_k^{(\pi)}  \right. \\
 & & \left. \times \mb{H}_k^H \left( \mb{I} + \rho \mb{H}_k \mb{Q}_k^{(\pi)''}
\mb{H}_k^H + \rho \lambda \mb{H}_k \Delta \mb{Q}_k^{(\pi)} \mb{H}_k^H  + \rho \ds{\sum_{\ell \in \mc{K}_{k}^{(\pi)}}}
\mb{H}_{\ell}\mb{Q}_{\ell}^{(\pi)}\mb{H}_{\ell}^H  \right)^{-1} \mb{H}_k \Delta \mb{Q}_k^{(\pi)} \right] \\
&= & - \mathbb{E} \mathrm{Tr} [\mb{A}  \Delta \mb{Q}_k^{(\pi)} \mb{A} \Delta \mb{Q}_k^{(\pi)}]
\end{array},
\]
with $\mb{A} = \rho^2 \mb{H}_k^H \left( \mb{I} +  \rho \mb{H}_k \mb{Q}_k^{(\pi)''}  \mb{H}_k^H + \rho \lambda \mb{H}_k \Delta \mb{Q}_k^{(\pi)}
\mb{H}_k^H  +\rho \ds{\sum_{\ell \in \mc{K}_{k}^{(\pi)}}} \mb{H}_{\ell}\mb{Q}_{\ell}^{(\pi)}\mb{H}_{\ell}^H  \right)^{-1}   \mb{H}_k $, which
can be proven to be a Hermitian positive definite matrix, $\Delta \mb{Q}_k^{(\pi)} = \mb{Q}_k^{(\pi)'} -\mb{Q}_k^{(\pi)''} $ also a Hermitian
matrix, and $\rho=\frac{1}{\sigma^2}$.
\[
\begin{array}{lcl}
\frac{\partial^2 f}{\partial \lambda^2}(\lambda) & = & - \mathbb{E} \mathrm{Tr} [\mb{A}^{1/2}   \Delta \mb{Q}_k^{(\pi)} \mb{A}^{1/2}
\mb{A}^{1/2} \Delta \mb{Q}_k^{(\pi)}
\mb{A}^{1/2} ] \\
& = & -\mathbb{E} \mathrm{Tr}[\mb{B}\mb{B}^H] < 0
\end{array},
\]
with $\mb{B}=\mb{A}^{1/2}   \Delta \mb{Q}_k^{(\pi)} \mb{A}^{1/2} $.

\subsection{Continuity of the utility functions $u_k^{\mathrm{SIC}}$}

Considering the Leibniz formula, the determinant of a matrix can be expressed as a weighted sum of products of its entries. Knowing that the
product and the sum of continuous functions are continuous, we conclude that the determinant function is continuous. Also, it is well known that
the logarithmic function is a continuous function. Thus, for any $\pi \in \mc{P}_K$, the function $R_k^{(\pi)}(\mb{Q}_k^{(\pi)},
\mb{Q}_{-k}^{(\pi)})$ is nothing else but the composition of two continuous functions which is also continuous w.r.t.
$(\mb{Q}_k^{(\pi)},\mb{Q}_{-k}^{(\pi)})$. This suffices to prove that $u_k^{\mathrm{SIC}}(\mb{Q}_k,\mb{Q}_{-k})$ is continuous w.r.t.
$(\mb{Q}_k,\mb{Q}_{-k})$.

\subsection{Convexity of the strategy sets $\mc{A}_k^{\mathrm{SIC}}$}

In order to prove that the set $\mc{A}_k^{\mathrm{SIC}}$ is convex, we need to verify that, for any two matrices $(\mb{Q}_k^{'}, \mb{Q}_k^{''})
\in \mc{A}_k^{\mathrm{SIC}} \times \mc{A}_k^{\mathrm{SIC}}$, we have:
$$
\alpha \mb{Q}_1^{'} +(1-\alpha) \mb{Q}_{1}^{''} \in \mc{A}_k^{\mathrm{SIC}},
$$
for all $\alpha \geq 0$.

For any $\mb{Q}_k^{'}, \mb{Q}_k^{''} \in \mc{A}_k^{\mathrm{(SIC)}}$, the matrices $\mb{Q}_k^{(\pi)}$ are Hermitian which implies that $\alpha
\mb{Q}_k^{(\pi)'} +(1-\alpha) \mb{Q}_{k}^{(\pi)''}$ are also Hermitian matrices, for all $\pi \in \mc{P}_K$.

Furthermore, for any $\mb{Q}_k^{'}, \mb{Q}_k^{''} \in \mc{A}_k^{\mathrm{SIC}}$, we have that $\mb{Q}_k^{(\pi)'}$, $\mb{Q}_k^{(\pi)''}$ are
non-negative matrices which implies that $\alpha \mb{Q}_k^{(\pi)'} +(1-\alpha) \mb{Q}_{k}^{(\pi)''}$ are also non-negative matrices, for all
$\pi \in \mc{P}_K$.

Finally, knowing that the trace is a linear application we have
that:
\[
\begin{array}{lll}
\ds{\sum_{\pi \in \mc{P}_k}} p_{\pi} \mathrm{Tr} \left(\alpha \mb{Q}_k^{(\pi)'} +(1-\alpha) \mb{Q}_{k}^{(\pi)''}\right) =  \\
= \alpha \ds{\sum_{\pi \in \mc{P}_k}} p_{\pi} \mathrm{Tr} (\mb{Q}_k^{(\pi)'})
 + (1-\alpha) \ds{\sum_{\pi \in \mc{P}_k}} p_{\pi} \mathrm{Tr} (\mb{Q}_k^{(\pi)'} ) \\
  \leq \alpha n_t \ol{P}_k +(1-\alpha) n_t \ol{P}_k  & & \\
 = n_t \ol{P}_k .
\end{array}
\]

Thus $\alpha \mb{Q}_k^{'} + (1-\alpha)\mb{Q}_{k}^{''} \in \mc{A}_k^{\mathrm{SIC}}$ and the set is convex.

\subsection{Compactness of the strategy sets $\mc{A}_k^{\mathrm{SIC}}$}

To prove that the strategy sets are compact sets we use the fact
that, in finite dimension spaces, a closed and bounded set is
compact.

First let us prove that $\mc{A}_k^{\mathrm{SIC}}$ is a closed set. We define the function $g : \mc{A}_k^{\mathrm{SIC}} \longrightarrow [0, n_t
\ol{P}_k]$, with
$$f(\mb{Q}_k)=\ds{\sum_{\pi \in \mc{P}_K}}p_{\pi} \mathrm{Tr}(\mb{Q}_k^{(\pi)}).$$

We see that $g(\cdot)$ is a continuous function and that its image is a compact and thus closed set. Knowing that the continuous inverse image
of a closed set is closed, we conclude that $\mc{A}_k^{\mathrm{SIC}}$ is closed.

Now we want to prove that the set $\mc{A}_k^{\mathrm{SIC}}$ is a bounded set. We associate to the tuple of matrices $(\mb{Q}_k^{(\pi)})_{\pi
\in \mc{P}_K}$ the following norm $||  \mb{Q}_k || = \ds{\sqrt{\sum_{\pi \in \mc{P}_K} ||  \mb{Q}_k^{(\pi)} ||^2_2}}$ where $||.||_2$ is is the
spectral norm of a matrix.
$$
||  \mb{Q}_k^{(\pi)} ||_2 = \sqrt{ \max\{ \lambda_{ \mb{Q}_k^{(\pi)H}\mb{Q}_k^{(\pi)}}(i)  \}_{i=1}^n }.
$$
Since for all $\mb{Q}_k \in \mc{A}_k^{\mathrm{SIC}}$ , $\mb{Q}_k^{(\pi)}$ is a non-negative, Hermitian matrix we have that:
$$
\max\{ \lambda_{  \mb{Q}_k^{(\pi)}}(i)  \}_{i=1}^n  \leq Tr(\mb{Q}_k^{(\pi)}) \leq \infty,
$$
and thus:
$$
||  \mb{Q}_k^{(\pi)} ||_2= \sqrt{ \max\{ \lambda_{  \mb{Q}_k^{(\pi)2} }(i)  \}_{i=1}^n  } =\sqrt{ \max\{ { \lambda_{  \mb{Q}_k^{(\pi)} }( i)
 }^2 \}_{i=1}^n  }\leq \infty.
$$

In conclusion the associated norm $||\mb{Q}_k||\leq \infty$.

\section{}
\label{appendix_2}

We suppose that there exist two different equilibrium strategy
profiles: $(\widetilde{\mb{Q}}_k,\widetilde{\mb{Q}}_{-k}) \in
\mc{A}_k^{\mathrm{SIC}} \times \mc{A}_{-k}^{\mathrm{SIC}} $ and $
(\widehat{\mb{Q}}_k,\widehat{\mb{Q}}_{-k}) \in
\mc{A}_k^{\mathrm{SIC}}\times\mc{A}_{-k}^{\mathrm{SIC}}$, such that
$(\widetilde{\mb{Q}}_k,\widetilde{\mb{Q}}_{-k}) \neq
(\widehat{\mb{Q}}_k,\widehat{\mb{Q}}_{-k})$.
 Then the condition given
in the theorem, $\mc{C} > 0$ is met for the particular choice of
$(\mb{Q}_k',\mb{Q}_{-k}')=(\widetilde{\mb{Q}}_k,\widetilde{\mb{Q}}_{-k})$
and
$(\mb{Q}_k'',\mb{Q}_{-k}'')=(\widehat{\mb{Q}}_k,\widehat{\mb{Q}}_{-k})$.

By the definition of the Nash Equilibrium, the strategies
$\widetilde{\mb{Q}}_k$, $k\in\mc{K}$, are the solutions of the
following maximization problems:

$$
\max_{\mb{Q}_k \in \mc{A}_k^{\mathrm{SIC}}}
u_k(\mb{Q}_k,\widetilde{\mb{Q}}_{-k}).
$$

Thus, $\widetilde{\mb{Q}}_k$ satisfy the following Kuhn-Tucker
optimality conditions:

\newcounter{saveenum_1}
\begin{enumerate}
\item{
$\widetilde{\mb{Q}}_k \in \mc{A}_k^{\mathrm{SIC}}$, which means
that:

\[  \left\{
\begin{array}{ccll}
\widetilde{\mb{Q}}_k^{(\pi)}  =  (\widetilde{\mb{Q}}_k^{(\pi)})^H & \succeq & 0 &,\forall \pi \in \mc{P}_K \\
\ds{\sum_{\pi \in \mc{P}_K}}
p_{\pi}\mathrm{Tr}(\widetilde{\mb{Q}}_k^{(\pi)})& \leq & n_t
\ol{P}_k,
\end{array}
\right. \]


}

\item{
There exist $\widetilde{\lambda}_k \geq 0$, and the following
Hermitian non-negative matrices of rank 1,
$\widetilde{\Phi}_k^{(\pi)}$, for all $\pi \in \mc{P}_K$, such that:

\[  \left\{
\begin{array}{ccll}
\widetilde{\lambda}_k\left[ \ds{\sum_{\pi \in \mc{P}_{K}}} p_{\pi} \mathrm{Tr}(\widetilde{\mb{Q}}_k^{(\pi)}) - n_t \ol{P}_k\right] & = & 0 \\
\mathrm{Tr}(\widetilde{\Phi}_k^{(\pi)} \widetilde{\mb{Q}}_k^{(\pi)})
& =& 0  &, \forall \pi \in \mc{P}_K,
\end{array}
\right. \]


}

\item{
\[
\left\{\begin{array}{lcl} \forall \pi \in \mc{P}_K :&&\\
\nabla_{\mb{Q}_k^{(\pi)}}u_k
(\widetilde{\mb{Q}}_k,\widetilde{\mb{Q}}_{-k}) & = &
p_{\pi}\widetilde{\lambda}_k \mb{I} - \widetilde{\Phi}_k^{(\pi)}
\end{array}\right. ,
\]

%

}

\setcounter{saveenum_1}{\value{enumi}}
\end{enumerate}

Having assumed that $(\widehat{\mb{Q}}_k,\widehat{\mb{Q}}_{-k})$ is
also a Nash Equilibrium, $\widehat{\mb{Q}}_k$, with $k \in \mc{K}$
are the solution of:
$$
\max_{\mb{Q}_k \in \mc{A}_k^{\mathrm{SIC}}}
u_k(\mb{Q}_k,\widehat{\mb{Q}}_{-k}),
$$
%
%
and thus $\widehat{\mb{Q}}_k$ satisfy the following Kuhn-Tucker
optimality conditions:

\begin{enumerate}
\setcounter{enumi}{\value{saveenum_1}}

\item{
$\widehat{\mb{Q}}_k \in \mc{A}_k^{\mathrm{SIC}}$, which means that:

\[  \left\{
\begin{array}{ccll}
\widehat{\mb{Q}}_k^{(\pi)} =  (\widehat{\mb{Q}}_k^{(\pi)})^H  & \succeq & 0 &,\forall \pi \in \mc{P}_K \\
\ds{\sum_{\pi \in \mc{P}_K}} p_{\pi}\mathrm{Tr}
(\widehat{\mb{Q}}_k^{(\pi)}) & \leq & n_t \ol{P}_k,
\end{array}
\right. \]

%
}

\item{
There exist $\widehat{\lambda}_k \geq 0$, $k \in \mc{K}$ and the
following non-negative, Hermitian matrices of rank 1,
$\widehat{\Phi}_k^{(\pi)}$, for all $\pi \in \mc{P}_K$ such that:

\[  \left\{
\begin{array}{ccll}
\widehat{\lambda}_k\left[\ds{\sum_{\pi \in \mc{P}_K}} p_{\pi} \mathrm{Tr}(\widehat{\mb{Q}}_k^{(\pi)}) - n_t \ol{P}_k \right] & = & 0 \\
\mathrm{Tr}(\widehat{\Phi}_k^{(\pi)} \widehat{\mb{Q}}_k^{(\pi)}) &
=& 0 &, \forall \pi \in \mc{P}_K,
\end{array}
\right. \]


}
\item{

\[
\left\{ \begin{array}{lcl} \forall \pi \in \mc{P}_K : \\
\nabla_{\mb{Q}_k^{(\pi)}}{u_k}
(\widehat{\mb{Q}}_k,\widehat{\mb{Q}}_{-k}) & = &
p_{\pi}\widehat{\lambda}_k \mb{I} - \widehat{\Phi}_k^{(\pi)}
\end{array}. \right.
\]


}

\end{enumerate}

Using the third and the sixth optimality conditions, the condition
given in (\ref{cond_ext_rosen_sic}) becomes:
\[
\begin{array}{lcl}
 \mc{C} &= &
 \ds{\sum_{\pi \in \mc{P}_K} \sum_{k=1}^K} \left\{  p_{\pi}\tilde{\lambda}_k \mathrm{Tr}(\widehat{\mb{Q}}^{(\pi)}_k)+p_{\pi}\hat{\lambda}_k \mathrm{Tr}(\widetilde{\mb{Q}}^{(\pi)}_k)-p_{\pi}\tilde{\lambda}_k
 \mathrm{Tr}(\widetilde{\mb{Q}}^{(\pi)}_k)-p_{\pi}\hat{\lambda}_k \mathrm{Tr}(\widehat{\mb{Q}}^{(\pi)}_k) -\right.\\
 & & \left. \mathrm{Tr}(\widehat{\mb{Q}}^{(\pi)}_k \widetilde{\mb{\Phi}}^{(\pi)}_k )
 - \mathrm{Tr}(\widetilde{\mb{Q}}^{(\pi)}_k \widehat{\mb{\Phi}}^{(\pi)}_k ) +  \mathrm{Tr}(\widetilde{\mb{Q}}^{(\pi)}_k \widetilde{\mb{\Phi}}^{(\pi)}_k ) +  \mathrm{Tr}(\widehat{\mb{Q}}^{(\pi)}_k \widehat{\mb{\Phi}}^{(\pi)}_k
 ) \right\}\\
& \leq  & \ds{\sum_{k=1}^K }\left\{\tilde{\lambda}_k
\left[\ds{\sum_{\pi \in \mc{P}_K}}p_{\pi}
\mathrm{Tr}(\widehat{\mb{Q}}^{(\pi)}_k)- n_t \ol{P}_k\right] +
\hat{\lambda}_k
\left[\ds{\sum_{\pi \in \mc{P}_K}}p_{\pi}\mathrm{Tr}(\widetilde{\mb{Q}}^{(\pi)}_k)-n_t\ol{P}_k\right] \right\}\\
& \leq & 0.
\end{array}
\]
From the other four K-T conditions, we obtain that all the terms on
the right are negative and thus $ \mc{C} \leq 0 $. But this
contradicts the diagonally strict concavity condition and so the
Nash Equilibrium is unique.

\section{}
\label{appendix_5}

We want to prove that there is no optimality loss when restricting
the search for the optimal covariance matrices to $\mb{Q}_k \in
\mc{A}_k^{\mathrm{SIC}}$ such that $\mb{Q}_k^{(\pi)}
=\mb{W}_k\mb{P}_k^{(\pi)}\mb{W}_k^H$, for all $\pi \in \mc{P}_K$.
Let us consider user $k\in \mc{K}$. We have that:
\begin{equation}
\begin{array}{lcl}
& & \ds{\arg \max_{\mb{Q}_k \in \mc{A}_k^{\mathrm{SIC}}} u_k(\mb{Q}_k,\mb{Q}_{-k})} \\
& = & \ds{ \arg \max_{\mb{Q}_k  \in \mc{A}_k^{\mathrm{SIC}}} \left\{
\ds{\sum_{\pi \in \mc{P}_K}} p_{\pi} \mathbb{E} \log_2 \left|\mb{I}
+ \rho \mb{H}_k \mb{Q}_k^{(\pi)}\mb{H}_k^H
 + \rho \ds{\sum_{\ell \in \mc{K}_{k}^{(\pi)}}} \mb{H}_{\ell} \mb{Q}_{\ell}^{(\pi)}\mb{H}_{\ell}^H \right| \right\} }\\
&= & \ds{ \arg \max_{\mb{Q}_k\in \mc{A}_k^{\mathrm{SIC}}}} \left\{
\ds{\sum_{\pi \in \mc{P}_K}} p_{\pi}\mathbb{E} \log_2 \left|\mb{I} +
\rho \mb{V}\tilde{\mb{H}}_k
\mb{W}_k^H\mb{Q}_k^{(\pi)}\mb{W}_k\tilde{\mb{H}}_k^H\mb{V}^H
 + \rho \ds{\sum_{\ell \in \mc{K}_{k}^{(\pi)}}} \mb{V}\tilde{\mb{H}}_{\ell}
\mb{W}_{\ell}^H\mb{Q}_{\ell}^{(\pi)}\mb{W}_{\ell}\tilde{\mb{H}}_{\ell}^H\mb{V}^H
\right|
\right\}\\
& = &\ds{ \arg \max_{\mb{Q}_k \in \mc{A}_k^{\mathrm{SIC}}}} \left\{
\ds{\sum_{\pi\in \mc{P}_K}} p_{\pi} \mathbb{E} \log_2 \left|\mb{I} +
\rho \tilde{\mb{H}}_k
\mb{W}_k^H\mb{Q}_k^{(\pi)}\mb{W}_k\tilde{\mb{H}}_k^H
 + \rho \ds{\sum_{\ell \in \mc{K}_{k}^{(\pi)}}}\tilde{\mb{H}}_{\ell}
\mb{W}_{\ell}^H\mb{Q}_{\ell}^{(\pi)}\mb{W}_{\ell}\tilde{\mb{H}}_{\ell}^H
\right|
\right\} \\
& = & \ds{\arg \max_{\mb{Q}_k\in \mc{A}_k^{\mathrm{SIC}}}} \left\{
\ds{\sum_{\pi \in \mc{K}_{k}^{(\pi)}}}\mathbb{E} \log_2 \left|\mb{I}
+ \rho \tilde{\mb{H}}_k \mb{X}_k^{(\pi)}\tilde{\mb{H}}_k^H + \rho
\ds{\sum_{\ell \in \mc{K}_{k}^{(\pi)}}}\tilde{\mb{H}}_{\ell}
\mb{W}_{\ell}^H\mb{Q}_{\ell}^{(\pi)}\mb{W}_{\ell}
\tilde{\mb{H}}_{\ell}^H \right| \right\}
\end{array},
\end{equation}
where we denoted with $\mb{X}_k^{(\pi)}\triangleq \mb{W}_k^H
\mb{Q}_k^{(\pi)}\mb{W}_k$. Knowing that the utility function is
concave w.r.t. the new defined matrices $\mb{X}_k^{(\pi)}$, and the
channel matrix $\mb{H}_k$ has independent entries, we can directly
apply the results given in \cite{tulino-wc-2006} to prove that
annulling the non-diagonal entries of $\mb{X}_k^{(\pi)}$ can only
increase the values of the functions $\mathbb{E} \log_2 \left|\mb{I}
+ \rho \tilde{\mb{H}}_k \mb{X}_k^{(\pi)}\tilde{\mb{H}}_k^H  + \rho
\ds{\sum_{\ell\in\mc{K}_{k}^{(\pi)}}}\tilde{\mb{H}}_{\ell}
\mb{W}_{\ell}^H\mb{Q}_{\ell}^{(\pi)}\mb{W}_{\ell}\tilde{\mb{H}}_{\ell}^H
\right|$. In conclusion the optimal matrices $\mb{X}_k^{(\pi)}$ are
diagonal, that we will denote with $\mb{P}_k^{(\pi)}$. The spectral
decomposition of the optimal covariance matrices are:
$\mb{Q}_k^{(\pi)}= \mb{W}_k \mb{P}_k^{(\pi)}\mb{W}_k^H$.


\begin{figure}
  \begin{center}
    \includegraphics[scale=0.50]{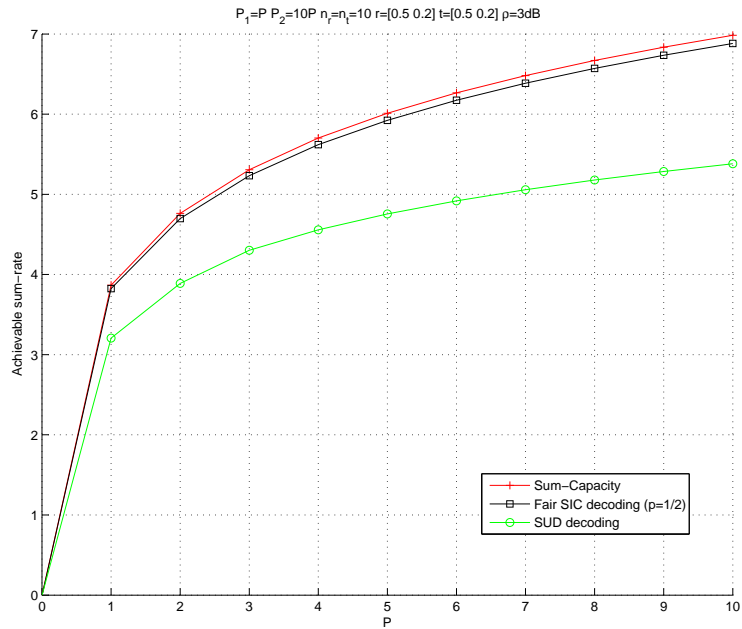}
  \end{center}
  \caption{ \footnotesize Fair SIC (joint space-time power allocation) vs. SUD decoding. Achievable network sum-rate versus the available transmit power $P$
 for $p=\frac{1}{2}$, $n_r=n_t=10$, $r=[0.5, 0.2]$, $t=[0.5, 0.2]$, $\rho=3 dB$. The fair SIC performs much closer to the sum-capacity upper bound than SUD.}
  \label{fig1}
\end{figure}

\begin{figure}
  \begin{center}
    \includegraphics[scale=0.50]{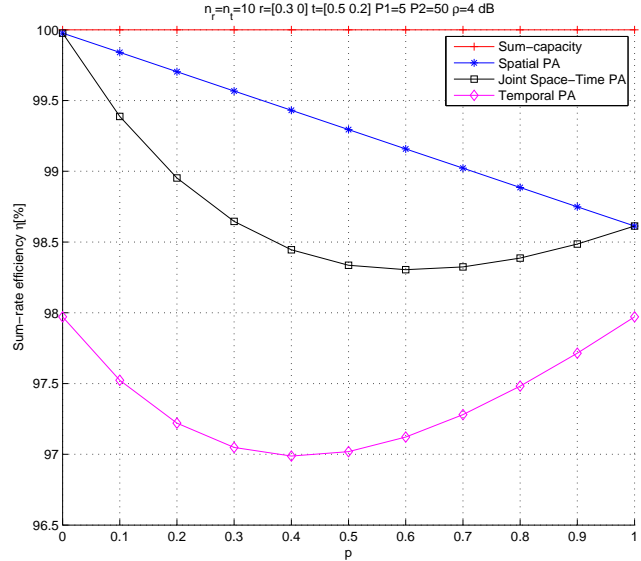}
  \end{center}
  \caption{ \footnotesize SIC decoding, comparaison between the joint space-time PA and the two special cases: the space PA and temporal PA. Sum-rate efficiency versus the distribution of the coordination signal $p \in [0,1]$
 for $n_r=n_t=10$, $r=[0.3, 0]$, $t=[0.5, 0.2]$, $\rho=4 dB$, $P_1=5$, $P_2=50$. The spatial PA outperforms the joint space-time PA (Braess paradox).}
  \label{fig2}
\end{figure}

\begin{figure}
  \begin{center}
    \includegraphics[scale=0.50]{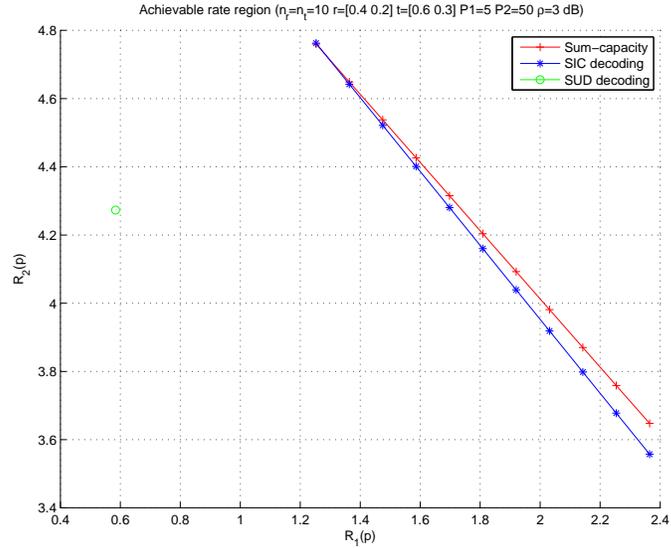}
  \end{center}
  \caption{ \footnotesize SIC decoding, space PA. The achievable rate region at the NE versus the distribution of the coordination signal $p \in [0,1]$
 for $n_r=n_t=10$, $r=[0.4, 0.2]$, $t=[0.6, 0.3]$, $\rho=3 dB$, $P_1=5$, $P_2=50$. Varying $p$ allows to move along a segment close to the sum-capacity, similar to SISO MAC.}
  \label{fig3}
\end{figure}


\bibliography{biblio}

\end{document}